\documentclass{article}

\usepackage{amsfonts, amsmath, amssymb, amsthm}

\newtheorem{theorem}{Theorem}

\theoremstyle{definition}

\theoremstyle{remark}

\author{Igor Korepanov}
\title{Two deformations of a fermionic solution to pentagon equation}

\begin{document}

\maketitle

\begin{abstract}
Two novel fermionic --- expressed in terms of Grassmann--Berezin calculus of anticommuting variables --- solutions of pentagon equation are proposed, both being deformations of the known solution related to the affine group.
\end{abstract}

\section{Introduction}\label{sec:P}

By pentagon equation we understand an algebraic relation which can be said to correspond naturally to a Pachner move $2\to 3$ --- an elementary rebuilding of a 3-manifold triangulation, which replaces two tetrahedra $1234$ and~$1235$ with three tetrahedra $1245$, $1345$ and~$2345$ occupying the same place in the manifold. If some quantities satisfy this relation, we say that a solution to pentagon equation has been found.

Many interesting solutions of pentagon equation are related to quantum invariants of 3-manifolds, and these were certainly a source of inspiration in our search. We are interested, however, in contrast to the usual approach, in solutions expressed in terms of Grassmann--Berezin calculus of anticommuting variables.

The particular kind of pentagon equation dealt with in this paper is~\eqref{5}, while its solutions are \eqref{G}, \eqref{g} and~\eqref{nR}. The already known solution is~\eqref{G}; there exists a conceptual proof of its validity and even a whole theory relating it to Reidemeister torsion of some unusual chain complexes associated with an affine group, see~\cite{bk}. It can also be expressed naturally as a fermionic Gaussian integral, we do it here in formula~\eqref{ef}.

Our new solutions in this paper are \eqref{g} and~\eqref{nR}, and they look even more unusual. Although we manage to represent them, too, in a form of Gaussian integral, their possible relations to Reidemeister torsions look obscure, and our proofs of their validity consist in direct calculations. Happily, these proofs are much simplified due to the fact that both \eqref{g} and~\eqref{nR} appear as ``perturbations'' of~\eqref{G}, namely by adding a new term of higher or lower degree, respectively, in anticommuting variables.

Below, in Section~\ref{sec:B} we recall the necessary things about the calculus of anticommuting variables; in Section~\ref{sec:G} we recall our ``old'' pentagon equation solution. Then we present our new solutions in Section~\ref{sec:g}, write out a Gaussian integral form for them in Section~\ref{sec:f} and finish with a brief discussion in Section~\ref{sec:d}.

\section{Grassmann algebras and Berezin integral}\label{sec:B}

A \emph{Grassmann algebra} over a field~$\mathbb F$ --- for which we can take in this paper any field of characteristic${}\ne 2$ --- is an associative algebra with unity, having generators~$a_i$ and relations
$$
a_i a_j = -a_j a_i .
$$
As this implies for $i=j$ that $a_i^2 =0$, any element of a Grassmann algebra is a polynomial of degree $\le 1$ in each~$a_i$. For a given Grassmann monomial, by its degree we understand its total degree in all Grassmann variables; if an element of Grassmann algebra includes only monomials of odd degrees, it is called odd; if it includes only monomials of even degrees, it is called even.

The \emph{exponent} is defined by the standard Taylor series. For example,
\[
\exp (a_1a_2) = 1+a_1a_2 .
\]
If $\varphi_1$ and~$\varphi_2$ are two even elements, then
\begin{equation}\label{ee}
\exp(\varphi_1) \exp(\varphi_2)=\exp(\varphi_1+\varphi_2).
\end{equation}

The \emph{Berezin integral}~\cite{B} is an $\mathbb F$-linear operator in a Grassmann algebra defined by equalities
\begin{equation}\label{iB}
\int \mathrm da_i =0, \quad \int a_i\, \mathrm da_i =1, \quad \int gh\, \mathrm da_i = g \int h\, \mathrm da_i,
\end{equation}
if $g$ does not depend on~$a_i$ (that is, generator~$a_i$ does not enter the expression for~$g$); multiple integral is understood as iterated one, according to the following model:
\begin{equation}\label{mB}
\iint ab\, \mathrm db\, \mathrm da = \int a \left( \int b\, \mathrm db \right) \mathrm da = 1 .
\end{equation}

\section{Solution of pentagon equation related to affine group}\label{sec:G}

\subsection{Tetrahedron weight}

We ascribe a \emph{coordinate} $\zeta_i \in \mathbb F$ to every vertex $i=1,\dots,5$ of tetrahedra taking part in the move $2\to 3$, see the first paragraph of Section~\ref{sec:P}. It will be also convenient to use the notation
\[
\zeta_{ij} \stackrel{\rm def}{=} \zeta_i - \zeta_j .
\]

Following paper~\cite{bk}\footnote{In this paper, we deal only with the ``scalar'' case of~\cite{bk}, not going into the more complicated matrix case.}, we attach anticommuting Grassmann generators to \emph{unoriented} 2-faces, such as $a_{123}=a_{132}=\dots=a_{321}$, and introduce the following function of coordinates and these generators --- the fermionic ``Boltzmann weight'' of a tetrahedron. To avoid bulky notations, we write it out for tetrahedron~$1234$; for another tetrahedron~$i_1i_2i_3i_4$ just change $k\mapsto i_k$, $k=1,\dots,4$:
\begin{multline}\label{G}
\mathbf f_{1234} = \frac{1}{\zeta_{34}} (\zeta_{23}a_{123}-\zeta_{24}a_{124}+\zeta_{34}a_{134})(\zeta_{13}a_{123}-\zeta_{14}a_{124}+\zeta_{34}a_{234}) \\
= \zeta_{12}a_{123}a_{124} - \zeta_{13}a_{123}a_{134} + \zeta_{14}a_{124}a_{134} \\
+ \zeta_{23}a_{123}a_{234} - \zeta_{24}a_{124}a_{234} + \zeta_{34}a_{134}a_{234} .
\end{multline}
Note that $\mathbf f_{1234}$ belongs to an oriented tetrahedron~$1234$, that is, it changes its sign under a change of orientation.

The weight $\mathbf f_{1234}$ is related to the group $\mathrm{Aff}(\mathbb F)$, i.e., the group of transformations of the form $x\mapsto xa+b$, but we do not explain it here, referring the reader to our paper~\cite{bk}.

\subsection{The pentagon equation}

As is known from~\cite{bk}, the following pentagon equation\footnote{In formula~\cite[(3)]{bk}, the convention~\eqref{mB} about the order of multiple integration was adopted. So, \cite[(3)]{bk} coincides essentially with our formula~\eqref{5}, we only interchanged $\mathrm da_{145} \leftrightarrow \mathrm da_{345}$ and wrote the minus sign arising from this. Then, however, there goes a slight confusion in that paper, because, starting from~\cite[Section~4]{bk}, a different convention was adopted.} holds for the $\mathbf f$'s defined by~\eqref{G}:
\begin{equation}\label{5}
\int \mathbf f_{1234} \mathbf f_{1235} \, \mathrm da_{123} = - \frac{1}{\zeta_{45}} \iiint \mathbf f_{1245} \mathbf f_{2345} \mathbf f_{1345} \, \mathrm da_{345}\, \mathrm da_{245}\, \mathrm da_{145} .
\end{equation}
See also Subsection~\ref{subs:e} below for some explanation of this.

\subsection{Relation to exponentials of bilinear forms}

Associate with tetrahedron~$1234$ the following matrix (which is to be compared with the expression between two equality signs in~\eqref{G}):
\begin{equation}\label{A}
A_{1234}=
\begin{pmatrix}
\zeta_{23} & -\zeta_{24} & \zeta_{34} & 0 \\
\zeta_{13}/\zeta_{34} & -\zeta_{14}/\zeta_{34} & 0 & 1
\end{pmatrix}
\end{equation}
and also two more Grassmann generators $b_{1234}^{(1)}$ and~$b_{1234}^{(2)}$. Consider the following bilinear form of Grassmann variables:
\begin{equation}\label{Phi}
\Phi_{1234} = \begin{pmatrix} b_{1234}^{(1)} & b_{1234}^{(2)} \end{pmatrix} A_{1234}
\begin{pmatrix} a_{123} \\ a_{124} \\ a_{134} \\ a_{234} \end{pmatrix}.
\end{equation}
Then it can be seen directly using~\eqref{iB} that the following Gaussian integral representation holds:
\begin{equation}\label{ef}
\mathbf f_{1234} = \iint \exp \Phi \;\mathrm db_{1234}^{(1)} \,\mathrm db_{1234}^{(2)}\, .
\end{equation}

Combining this with the property~\eqref{ee}, we see that both sides in~\eqref{5} can be expressed as multiple (five-fold in the l.h.s.\ and nine-fold in the r.h.s.) integrals of bilinear forms.

\subsection{Some explicit expressions}\label{subs:e}

As this paper is about direct calculations, it makes sense to write out here the matrices of bilinear forms corresponding to the l.h.s.\ and r.h.s.\ of~\eqref{5}. The building blocks for them are copies of matrix~\eqref{A}.

The matrix for l.h.s.\ is 
\begin{equation}\label{lG}
\begin{pmatrix}
\zeta_{23} & -\zeta_{24} & 0 & \zeta_{34} & 0 & 0 & 0 \\
\zeta_{13}/\zeta_{34} & -\zeta_{14}/\zeta_{34} & 0 & 0 & 0 & 1 & 0 \\
\zeta_{23} & 0 & -\zeta_{25} & 0 & \zeta_{35} & 0 & 0 \\
\zeta_{13}/\zeta_{35} & 0 & -\zeta_{15}/\zeta_{35} & 0 & 0 & 0 & 1
\end{pmatrix} ;
\end{equation}
the rows correspond to $b_{1234}^{(1)}$, $b_{1234}^{(2)}$, $b_{1235}^{(1)}$, and~$b_{1235}^{(2)}$; the columns correspond to $a_{123}$, $a_{124}$, $a_{125}$, $a_{134}$, $a_{135}$, $a_{234}$, and~$a_{235}$.

The matrix for r.h.s.\ is 
\begin{equation}\label{rG}
\begin{pmatrix}
\zeta_{24} & -\zeta_{25} & 0 & 0 & \zeta_{45} & 0 & 0 & 0 & 0 \\
\zeta_{14}/\zeta_{45} & -\zeta_{15}/\zeta_{45} & 0 & 0 & 0 & 0 & 0 & 1 & 0 \\
0 & 0 & \zeta_{34} & -\zeta_{35} & \zeta_{45} & 0 & 0 & 0 & 0 \\
0 & 0 & \zeta_{14}/\zeta_{45} & -\zeta_{15}/\zeta_{45} & 0 & 0 & 0 & 0 & 1 \\
0 & 0 & 0 & 0 & 0 & \zeta_{34} & -\zeta_{35} & \zeta_{45} & 0 \\
0 & 0 & 0 & 0 & 0 & \zeta_{24}/\zeta_{45} & -\zeta_{25}/\zeta_{45} & 0 & 1 
\end{pmatrix} ;
\end{equation}
the rows correspond to $b_{1245}^{(1)}$, $b_{1245}^{(2)}$, $b_{1345}^{(1)}$, $b_{1345}^{(2)}$, $b_{2345}^{(1)}$, and~$b_{2345}^{(2)}$; the columns correspond to  $a_{124}$, $a_{125}$, $a_{134}$, $a_{135}$, $a_{145}$, $a_{234}$, $a_{235}$, $a_{245}$, and~$a_{345}$.

The coefficient at every Grassmann monomial in a bilinear form is the minor of its matrix standing in the intersection of the rows and columns corresponding to the variables in that monomial. This reduces the proof of~\eqref{5} to comparing minors of matrices \eqref{lG} and~\eqref{rG}; such minor must include all rows of the corresponding matrix (because all the~$b$'s must be integrated out) and those columns corresponding to inner faces (for the same reason; the inner faces are, of course, $123$ in the l.h.s., and $145$, $245$, and~$345$ in the r.h.s.); other columns must correspond to the same~$a$'s in \eqref{lG} and~\eqref{rG}. Also, the signs must be taken into account appearing when we bring a variable to the right in order to integrate it out, as well as the factor~$(-1/\zeta_{45})$ in~\eqref{5}.

Fortunately, there exists a theory saving us from actually doing all these calculations, because of a proportionality of the mentioned minors; this is explained in the proof of Theorem~3 in paper~\cite{bk}. In fact, just one pair of minors must be compared.

\section{New solutions}\label{sec:g}

\subsection{Solution with term of degree~4}\label{subs:nG}

We add one more term --- found by method of free search and trial --- to \eqref{G}:
\begin{equation}\label{g}
\mathbf g_{1234} \stackrel{\rm def}{=} \mathbf f_{1234} + \epsilon_{1234}\, \lambda\, c_{1234}\, a_{123}a_{124}a_{134}a_{234},
\end{equation}
and similarly, with substitution $k\mapsto i_k$, for any tetrahedron~$i_1i_2i_3i_4$. In~\eqref{g},
\[
c_{1234} = \prod_{1\le i<j\le 4} \zeta_{ij} ,
\]
$\lambda$ is an overall parameter, and $\epsilon_{1234}$ is simply the unity, but in general $\epsilon_{i_1i_2i_3i_4}=\pm 1$: if the tetrahedron orientation determined by the order of vertices~$i_1,i_2,i_3,i_4$ is \emph{consistent} with the orientation~$1,2,3,4$ for tetrahedron~$1234$, then it is~$1$, otherwise~$-1$. Even more directly: $\epsilon_{1235}=-1$, $\epsilon_{1245}=-1$, $\epsilon_{1345}=1$, and~$\epsilon_{2345}=-1$.

\begin{theorem}\label{th:g}
The $\mathbf g$'s defined by~\eqref{g} satisfy the same pentagon equation as the~$\mathbf f$'s, i.e.,
\begin{equation}\label{5g}
\int \mathbf g_{1234} \mathbf g_{1235} \, \mathrm da_{123} = - \frac{1}{\zeta_{45}} \iiint \mathbf g_{1245} \mathbf g_{2345} \mathbf g_{1345} \, \mathrm da_{345}\, \mathrm da_{245}\, \mathrm da_{145} .
\end{equation}
\end{theorem}

\begin{proof}[Sketch of the proof]
The only known to us proof of Theorem~\ref{th:g} consists in direct calculations. These are simplified by
\begin{enumerate}
\item\label{i:f} the fact that the~$\mathbf f$'s already satisfy~\eqref{5},
\item\label{i:a} the fact that $a^2=0$ for a Grassmann generator~$a$, and
\item\label{i:s} the symmetries of~\eqref{5g}: it transforms into itself under any permutation of vertices~$1,2,3$, as well as~$4,5$.
\end{enumerate}
It follows from~\ref{i:f} that all monomials of degree~3 in the l.h.s.\ and r.h.s.\ of~\eqref{5g} are already the same. Due to~\ref{i:a}, only monomials of degree~5 remain to be checked, and \ref{i:s}~makes it enough to check the coefficients at just one monomial of degree~5 in both sides of~\eqref{5g}, for instance, the factors at $a_{124}a_{125}a_{134}a_{135}a_{235}$. This has been actually done first using paper and pencil and then double-checked using GAP computer algebra system~\cite{GAP}.
\end{proof}

\subsection{Solution with term of degree~0}

There is also a somewhat similar but simpler solution of pentagon equation:
\begin{equation}\label{nR}
\mathbf h_{1234} \stackrel{\rm def}{=} \mathbf f_{1234} + \epsilon_{1234}\, \mu,
\end{equation}
and similarly for other tetrahedra. Here $\mu$ is an overall parameter, and $\epsilon_{i_1i_2i_3i_4}$ has the same meaning as in Subsection~\ref{subs:nG}.

\begin{theorem}\label{th:nR}
The $\mathbf h$'s defined by~\eqref{nR} satisfy the same pentagon equation as the~$\mathbf f$'s and~$\mathbf g$'s, namely,
\begin{equation}\label{5h}
\int \mathbf h_{1234} \mathbf h_{1235} \, \mathrm da_{123} = - \frac{1}{\zeta_{45}} \iiint \mathbf h_{1245} \mathbf h_{2345} \mathbf h_{1345} \, \mathrm da_{345}\, \mathrm da_{245}\, \mathrm da_{145} .
\end{equation}
\end{theorem}

\begin{proof}[Sketch of the proof]
Again, the only known to us proof of Theorem~\ref{th:nR} consists in direct calculations. The difference with Theorem~\ref{th:g} is that here we must check a monomial of degree~1, instead of~5.
\end{proof}

\section{Representing new solutions as Gaussian integrals}\label{sec:f}

\subsection{Gaussian integral for~$\mathbf g$}

It is not difficult to see directly that our solution~\eqref{g} can be written in the following Gaussian integral form: replace~$\Phi_{1234}$ given by~\eqref{Phi} with
\[
\Gamma_{1234} = \Phi_{1234} + \epsilon_{1234}\lambda_{1234}\zeta_{13}\zeta_{14}\zeta_{23}\zeta_{24}\zeta_{34}a_{134}a_{234},
\]
then the analogue of~\eqref{ef} holds:
\[
\mathbf g_{1234} = \iint \exp \Gamma \; \mathrm db_{1234}^{(1)}\, \mathrm db_{1234}^{(2)}.
\]

\subsection{Gaussian integral for~$\mathbf h$}

Neither is difficult to bring~\eqref{nR} to the Gaussian form: replace~$\Phi_{1234}$ given by~\eqref{Phi} with
\[
\Psi_{1234} = \Phi_{1234} + b_{1234}^{(2)}b_{1234}^{(1)},
\]
then
\[
\mathbf h_{1234} = \iint \exp \Psi \; \mathrm db_{1234}^{(1)}\, \mathrm db_{1234}^{(2)}.
\]

\subsection{$\Gamma$ and~$\Psi$ not as simple as~$\Phi$}\label{subs:p}

One big new feature of forms $\Gamma$ and~$\Psi$, compared to~$\Phi$, is that neither $\Gamma$ nor~$\Psi$ is any longer a form linear, separately, in~$a$'s belonging to 2-faces, on one hand, and $b$'s belonging to tetrahedra, on the other hand. This makes it problematic to associate with $\Gamma$ or~$\Psi$, at least in a direct way, a matrix whose copies could be used, first, as building blocks for a larger matrix (like, for a simple instance, \eqref{lG} or~\eqref{rG}), and then include this larger matrix in a sequence of matrices forming a chain complex. Recall that in~\cite{bk} and our other papers, the Reidemeister torsion of a complex built in such way was used to construct manifold invariants.

\section{Discussion}\label{sec:d}

Here are some concluding remarks:

\begin{itemize}
\item The most intriguing thing about our solutions~\eqref{g} and~\eqref{nR} is that their ``mother solution''~\eqref{G} has a four-dimensional generalization~\cite{4} and, in fact, generalizes to any manifold dimension~\cite{ks}. So, it may make sense to search for higher dimensional generalizations of \eqref{g} and~\eqref{nR} as well. This search may be started with \emph{infinitesimal} perturbations of the known solutions: if they exist, this will be already of great interest.
\item As we already mentioned, the solution~\eqref{G} is known~\cite{bk} to be closely related to the group~$\mathrm{Aff}(\mathbb F)$. At this moment, it is unclear whether this relation is conserved for our new solutions, or maybe $\mathrm{Aff}(\mathbb F)$ should be replaces by another algebraic object.
\item Also, Subsection~\ref{subs:p} suggests that some generalization of Reidemeister torsion may be needed.
\item Of course, the behavior of our solutions with respect to Pachner moves $1\to 4$ (a tetrahedron is divided in four, so that a new vertex appears within it) deserves close attention. After obtaining necessary formulas, we can look at what kind of manifold invariants this brings about.
\item We could not (as yet?) unite \eqref{g} and~\eqref{nR} somehow into one ``composite'' solution.
\end{itemize}

\end{document}